\documentclass[11pt]{article}
\usepackage{amssymb}
\usepackage[perpage,symbol]{footmisc}
\usepackage{listings}
\usepackage{amsfonts}
\usepackage{enumerate}
\usepackage{amsmath}
\usepackage{cite}
\usepackage{array}
\usepackage{booktabs}

\begin{document}

\newtheorem{theorem}{Theorem}[section]
\newtheorem{corollary}[theorem]{Corollary}
\newtheorem{definition}[theorem]{Definition}
\newtheorem{proposition}[theorem]{Proposition}
\newtheorem{lemma}[theorem]{Lemma}
\newtheorem{example}[theorem]{Example}
\newtheorem{algorithm}[theorem]{Algorithm}
\newenvironment{proof}{\noindent {\bf Proof.}}{\rule{3mm}{3mm}\par\medskip}
\newcommand{\remark}{\medskip\par\noindent {\bf Remark.~~}}
\title{An NP-hard generalization of Nim}
\author{Chunlei Liu\footnote{Shanghai Dengbi Communication Tech. Co., Shanghai 200240, clliu@sjtu.edu.cn.}}
\date{}
\maketitle
\thispagestyle{empty}

\abstract{A new combinatorial game is given. It generalizes both Substraction and Nim. It is proved the computation of the Nash equilibrium points in these new games is NP-hard.}

\noindent {\bf Key words}:  combinatorial game,  Nash equilibrium point, complexity, NP-complete, NP-hard.

\section{\small{Introduction}}
\hskip .2in
In this section we introduce the notion of occupation games.
\begin{definition}Let $X$ be a  finite set, $S$ a set of subsets of $X$,
$O$ a set of nonempty subsets of $X$, and
$A\in S$. The $S$-admissible restriction of $O$ on $A$ is the set
  $$O\mid_{A,S}=\{\sigma\in O\mid \sigma\subseteq A,A-\sigma\in S \}.$$\end{definition}
\begin{definition}Let $X$ be a  finite set, $S$ a set of subsets of $X$, and  $O$ a set of nonempty subsets of $S$. Then the triple $(X,S,O)$, together with the system of move laws:
$$A\stackrel{\sigma}{\rightarrow}A-\sigma,\ \forall A\in S,\ \forall \sigma\in O\mid_{A,S},$$
is called an occupation game.
Starting from $A_0\in S$, two players take turns to
get a sequence of moves
$$A_0\stackrel{\sigma_0}{\rightarrow}
A\stackrel{\sigma_1}{\rightarrow}
\cdots\stackrel{\sigma_{n-1}}{\rightarrow}A_n$$
until $O\mid_{A_n,S}=\emptyset$. The player whose turn is to move $A_n$ is the loser.
\end{definition}
\begin{example} Let  $X=\sqcup_{i=1}^nX_i$ be a disjoint unions of finite sets, $S$ the set of subsets of $X$, and  $O=\sqcup_{i=1}^nS_i$, where $S_i$ is the set of 1-subsets and 2-subsets of
$X_i$.
Then the occupation game $(X,S,O)$ is Substraction with $n$-piles.
\end{example}
\begin{example}[\cite{Bou02},\cite{Gru39}] Let  $X=\sqcup_{i=1}^nX_i$ be a disjoint unions of finite sets, $S$ the set of subsets of $X$, and
 and  $O=\sqcup_{i=1}^nS_i$, where $S_i$ is the set of nonempty subsets   of
$X_i$.
Then $(X,S,O)$ is  Nim with $n$ piles.
\end{example}
From the second example, an occupation game is a generalization of Nim. For other interesting generalizations of Nim, please consult \cite{BCG1} and \cite{BCG2}.
Let $(X,S,O)$ be an occupation  game, and $A\in S$.
Then the Nash equilibrium point of  $A$ is $0$ or $1$. So we denote it by  ${\rm Truth}(A)$.
By \cite{Na50}, if ${\rm Truth}(A)=1$, then the first player of $A$ has a strategy to win. Similarly, if ${\rm Truth}(A)=0$, then the second player of $A$ has a strategy to win. It follows that
\begin{itemize}
  \item ${\rm Truth}(A)=0$ if  $A=\emptyset$,
  \item ${\rm Truth}(A)=1$ if ${\rm Truth}(A-\sigma)=0$ for some  $\sigma\in O\mid_{A,S}$, and
  \item ${\rm Truth}(A)=0$ if ${\rm Truth}(A-\sigma)=1$ for all $\sigma\in O\mid_{A,S}$.
\end{itemize}
Given an occupation  game  $(X,S,O)$, and $A\in S$,
the computation of ${\rm Truth}(A)$ is  interesting.
\begin{example}Let $(X,S,O)$ be $n$-piles Substraction  with $X=\sqcup_{i=1}^nX_i$. For each $i$, let  $A_i\subseteq X_i$, $a_i=0,1,2$ be the remainder of $|A_i|$ modulo $3$,
  and 
  $a_i=a_{i0}+a_{i1}2$ with $a_{i0},a_{i1}=0,1$.
Then
$${\rm Truth}(\cup_{i=1}^nA_i)\oplus1=\prod_{j=0}^1(1\oplus\oplus_{i=1}^na_{ij}).$$
\end{example}
\begin{example}[\cite{Bou02}]Let $(X,S,O)$ be $n$-piles Nim with  $X=\sqcup_{i=1}^nX_i$.
For each $i$, let   $A_i\subseteq X_i$, $|A_i|=\sum_{j=0}^ka_{ij}2^j$  with $a_{ij}={0,1}$.
Then $${\rm Truth}(\cup_{i=1}^nA_i)\oplus1=\prod_{j=0}^k(1\oplus\oplus_{i=1}^na_{ij}).$$
\end{example}

In general, the computation of ${\rm Truth}(A)$ in an occupation game is not easy. So, according to from \cite{Co71}, and \cite{Ka72}, it is interesting to know the complexity of the computation of ${\rm Truth}(A)$
we prove the following theorem.
\begin{theorem}\label{hard}
Let $(X,S,O)$ be an occupation  game and $A\in S$. Then the computation of ${\rm Truth}(A)$  is ${\rm NP}$-hard. \end{theorem}
\section{${\rm NP}$-hardness}
In this section we prove Theorem \ref{hard}.
\begin{definition}
For a tuple  $(t_1,\cdots,t_n,t)$ of numbers, the corresponding occupation game $(X,S,O)$ is constructed as follows:
\begin{itemize}
  \item $X=(V\sqcup W)\sqcup L\sqcup Y$, where $|V|=|W|=n$, $|L|=2nt+n-1$, $Y=\sqcup_{i=1}^nY_i$ with $|Y_i|=t_i$.
  \item $S$ is the set of subsets $A$ of $X$ such that
$$0\leq |A\cap W|-|A\cap V|\leq1.$$
  \item $O=O_1\cup O_2$, where
$O_1$ is the set of subsets $\sigma$ of $X$ such that
$$\sigma\cap W=\emptyset,$$
and $$|\sigma\cap V|=1,\ \sigma\cap Y\in\{Y_1,\cdots,Y_n\},\ \ 
|\sigma\cap L|\in\{2n|\sigma\cap Y|,0\},$$
and where
$O_2$ is the set of subsets $\sigma$ of $X$ such that
$$\sigma\cap(Y\cup V)=\emptyset,$$
and $$|\sigma\cap W|=1,\   |\sigma\cap L|=1.$$
\end{itemize}
\end{definition}
\begin{lemma}
Let $(t_1,\cdots,t_n,t)$ be tuple of numbers, and $(X,S,O)$ the corresponding occupation game  constructed in the last definition.
Let $A\in S$.
Then the following statements are true.
\begin{itemize}
  \item If
  $|A\cap V|\geq|A\cap W|$, then $O_2\mid_{A,S}=\emptyset$.
  \item If
  $|A\cap V|< |A\cap W|$, then $O_1\mid_{A,S}=\emptyset$.
\end{itemize}
\end{lemma}
\begin{proof}
Suppose that $|A\cap V|\geq|A\cap W|$.
  Let $\sigma\in O_2$ with $\sigma\subseteq A$. Then 
  $$|(A-\sigma)\cap V|=|A\cap V|,$$
  and
  $$ |(A-\sigma)\cap W|=|A\cap W|-1\leq |(A-\sigma)\cap V|-1.$$
  Thus $A-\sigma\not\in S$.
Therefore  $O_2\mid_{A,S}=\emptyset$.
The first item is proved. The second item can be proved similarly.
\end{proof}
\begin{lemma}
Let $(t_1,\cdots,t_n,t)$ be tuple of numbers, and $(X,S,O)$ the corresponding occupation game  constructed in the last definition.
Suppose that there is a  subset $I$ of $\{1,\cdots,n\}$ such that
 $\cup_{i\in I}Y_i$ has $t$ elements.
Let $A\in S$.
Then the following statements are true.
  \begin{itemize}
    \item 
  If $J\subseteq I$ or $I\subseteq J$,
    $|A\cap V|=|A\cap W|-1=n-|J|$,
  $A\cap Y=\cup_{j\notin J} Y_j$,
  $$|A\cap L|=2n(t-\sum_{j\in J\cap I}t_j)+n-|J|,$$
  then ${\rm Truth}(A)=0$, and
    \item
  If $J\varsubsetneq I$ or $I\subseteq J\varsubsetneq\{1,\cdots,n\}$, $|A\cap V|=|A\cap W|=n-|J|$,
  $A\cap Y=\cup_{j\notin J} Y_j$,
  $$|A\cap L|=2n(t-\sum_{j\in J\cap I}t_j)+n-1-|J|,$$
  then ${\rm Truth}(A)=1$.
  \end{itemize} 
\end{lemma}
\begin{proof}
Firstly, we suppose that 
$J=\{1,\cdots,n\}$, and $A$ satisfies the conditions in the first item.
Then $|A\cap V|=|A\cap W|=0$,
  $|A\cap L|=0$.
  Let $\sigma\in O\mid_{A,S}$.
  By the last lemma, $\sigma\in O_2\mid_{A,S}$.
  So $|\sigma\cap L|=1$, 
  contradicting to
  $|A\cap L|=0$.
  Therefore
  $O\mid_{A,S}=\emptyset$, and hence  ${\rm Truth}(A)=0$.
  
Secondly, we suppose that   $I\subseteq J\varsubsetneq\{1,\cdots,n\}$,  and $A$ satisfies the conditions in the first item.
 Then $$|A\cap V|=|A\cap W|-1=n-|J|,\ 
  A\cap Y=\cup_{j\notin J} Y_j,\ |A\cap L|=n-|J|.$$
  Let $\sigma$ be an arbitrary element in $O\mid_{A,S}$.
  By the last lemma, $\sigma\in O_2\mid_{A,S}$. So
  $$\sigma\cap (Y\cup V)=\emptyset,\
  |\sigma\cap W|=1,\  |\sigma\cap L|=1.$$
  Then $I\subseteq J$,
  $$|(A-{\sigma})\cap V|=n-|J|,$$
  $$|(A-{\sigma})\cap W|=n-|J|,$$
  $$|(A-{\sigma})\cap Y|=\cup_{j\notin J}Y_j,$$
  $$|(A-{\sigma})\cap L|=n-1-|J|.$$
  Thus $A-{\sigma}$  satisfies conditions in the second item. By induction,
  ${\rm Truth}(A-{\sigma})=1$. Since $\sigma$ is arbitrary chosen,
  we have ${\rm Truth}(A)=0$.
  
Thirdly, we suppose that   $ J\varsubsetneq I$, and $A$ satisfies the conditions in the first item.
Then
  $$|A\cap V|=|A\cap W|-1=n-|J|,\
  A\cap Y=\cup_{j\notin J} Y_j,$$
   $$|A\cap L|=2n(t-\sum_{j\in J}t_j)+n-|J|.$$
  Let $\sigma$ be an arbitrary element in $O\mid_{A,S}$.
  By the last lemma, $\sigma\in O_2\mid_{A,S}$. So
  $$\sigma\cap (Y\cup V)=\emptyset,\
  |\sigma\cap W|=1,\  |\sigma\cap L|=1.$$
  Then $I\subseteq J$,
  $$|(A-{\sigma})\cap V|=n-|J|,$$
  $$|(A-{\sigma})\cap W|=n-|J|,$$
  $$|(A-{\sigma})\cap Y|=\cup_{j\notin J}Y_j,$$
  $$|(A-{\sigma})\cap L|=2n(t-\sum_{j\in J}t_j)+n-1-|J|.$$
  Thus $A-{\sigma}$  satisfies conditions in the second item. By induction,
  ${\rm Truth}(A-{\sigma})=1$. Since $\sigma$ is arbitrary chosen,
  we have ${\rm Truth}(A)=0$.

Fourthly, we suppose that   $I\subseteq J\varsubsetneq\{1,\cdots,n\}$, and $A$ satisfies the conditions in the second item.
Then
  $$|A\cap V|=|A\cap W|=n-|J|,\
  A\cap Y=\cup_{j\notin J} Y_j,\ |A\cap L|=n-1-|J|.$$
  Take $j_0\notin J$, take $\sigma\in O_1\mid_{A,S}$ to be one such that
  $$\sigma\cap W=\emptyset,\
  |\sigma\cap V|=1,\ \sigma\cap Y=Y_{j_0},\ |\sigma\cap L|=0.$$
  Then $I\subseteq J\cup\{j_0\}$,
  $$|(A-{\sigma})\cap V|=n-(|J|+1)=n-|J\cup\{j_0\}|,$$
  $$|(A-{\sigma})\cap W|=n-|J|=n-|J\cup\{j_0\}|+1,$$
  $$|(A-{\sigma})\cap Y|=\cup_{j\notin J\cup\{j_0\}}Y_j,$$
  $$|(A-{\sigma})\cap L|=n-1-|J|=n-|J\cup\{j_0\}|.$$
  Thus $A-{\sigma}$  satisfies conditions the first item with  $J$ replaced by $J\cup\{j_0\}$. By induction,
  ${\rm Truth}(A-{\sigma})=0$.
  Therefore, ${\rm Truth}(A)=1$.

Finally, we suppose that   $J\varsubsetneq I$, and $A$ satisfies the conditions in the second item.
Then
  $$|A\cap V|=|A\cap W|=n-|J|,\
  A\cap Y=\cup_{j\notin J} Y_j,$$
  $$|A\cap L|=2n(t-\sum_{j\in J}t_j)+n-1-|J|.$$
  Take $j_0\in I-J$, take $\sigma\in O_1\mid_{A,S}$ to be one such that
  $$\sigma\cap W=\emptyset,\
  |\sigma\cap V|=1,\ \sigma\cap Y=Y_{j_0},\ |\sigma\cap L|=2nt_{j_0}.$$
  Then $ J\cup\{j_0\}\subseteq I$,
  $$|(A-{\sigma})\cap V|=n-(|J|+1)=n-|J\cup\{j_0\}|,$$
  $$|(A-{\sigma})\cap W|=n-|J|=n-|J\cup\{j_0\}|+1,$$
  $$|(A-{\sigma})\cap Y|=\cup_{j\notin J\cup\{j_0\}}Y_j,$$
  $$|(A-{\sigma})\cap L|=2n(t-\sum_{j\in J\cup \{j_0\}}t_j)+n-1-|J|=2n(t-\sum_{j\in J\cup \{j_0\}}t_j)+n-|J\cup\{j_0\}|.$$
  Thus $A-{\sigma}$  satisfies conditions the first item with  $J$ replaced by $J\cup\{j_0\}$. By induction,
  ${\rm Truth}(A-{\sigma})=0$.
  Therefore, ${\rm Truth}(A)=1$.
  The proof of the lemma is completed.
\end{proof}

\begin{lemma}
Let $(t_1,\cdots,t_n,t)$ be tuple of numbers, and $(X,S,O)$ the corresponding occupation game  constructed in the last definition.
Suppose that there is no  subset $I$ of $\{1,\cdots,n\}$ such that
 $\cup_{i\in I}Y_i$ has $t$ elements.
Let $A\in S$.
Then the following statements are true.
  \begin{itemize}
    \item
  If     $|A\cap V|=|A\cap W|-1=n-|J|$, $A\cap Y=\cup_{j\notin J} Y_j$,
  $$|A\cap L|=2n(t-\sum_{j\in J_1}t_j)+n-|J|,\ \sum_{j\in J_1}t_j<t,$$
  then ${\rm Truth}(A)=0$, and
    \item
  If  $|A\cap V|=|A\cap W|=n-|J|$, $A\cap Y=\cup_{j\notin J} Y_j$,
  $$|A\cap L|=2n(t-\sum_{j\in J'}t_j)+n-1-|J|,\ \sum_{j\in J'}t_j<t,$$
  then ${\rm Truth}(A)=0$.
  \end{itemize}
\end{lemma}
\begin{proof}
Firstly, we suppose that
$J=\{1,\cdots,n\}$, and $A$ satisfies the conditions in the second item.
Then
  $|A\cap V|=|A\cap W|=0$.
  Let $\sigma\in O\mid_{A,S}$.
  By a previous lemma, $\sigma\in O_1\mid_{A,S}$.
  So $|\sigma\cap V|=1$,
  contradicting to
  $|A\cap L|=0$.
  Therefore
  $O\mid_{A,S}=\emptyset$, and hence  ${\rm Truth}(A)=0$.

Secondly, we suppose that   $J\varsubsetneq\{1,\cdots,n\}$, and $A$ satisfies the conditions in the second item.
Then
  $$|A\cap V|=|A\cap W|=n-|J|,\
  A\cap Y=\cup_{j\notin J} Y_j,$$
  $$|A\cap L|=2n(t-\sum_{j\in J'}t_j)+n-1-|J|,\ \sum_{j\in J'}t_j<t.$$
  Let $\sigma$ be an arbitrary element in $O\mid_{A,S}$.
  By a previous lemma, $\sigma\in O_1\mid_{A,S}$. So
  $$\sigma\cap W=\emptyset,\
  |\sigma\cap V|=1,\ \sigma\cap Y=\{Y_{j_0}\}, j_0\notin J$$
  Then 
  $$|(A-{\sigma})\cap V|=n-|J|-1=n-|J\cup\{j_0\}|,$$
  $$|(A-{\sigma})\cap W|=n-|J|=n-|J\cup\{j_0\}|+1,$$
  $$|(A-{\sigma})\cap Y|=\cup_{j\notin J}Y_j,$$
  $$|A\cap L|=2n(t-\sum_{j\in J_1}t_j)+n-1-|J|,\ \sum_{j\in J_1}t_j<t.$$
  Thus $A-{\sigma}$  satisfies conditions in the first item. By induction,
  ${\rm Truth}(A-{\sigma})=1$. Since $\sigma$ is arbitrary chosen,
  we have ${\rm Truth}(A)=0$.

Thirdly, we suppose that $A$  satisfies conditions in the first item. 
Then
  $$|A\cap V|=n-|J|,$$
  $$|A\cap W|=n-|J|+1,$$
  $$|A\cap Y|=\cup_{j\notin J}Y_j,$$
  $$|A\cap L|=2n(t-\sum_{j\in J_1}t_j)+n-1-|J|,\ \sum_{j\in J_1}t_j<t.$$
  Let $\sigma$ be an arbitrary element in $O\mid_{A,S}$.
  By a previous lemma, $\sigma\in O_2\mid_{A,S}$. So
  $$\sigma\cap (Y\cup V)=\emptyset,\
  |\sigma\cap W|=1,\  |\sigma\cap L|=1.$$
  Then $I\subseteq J$,
  $$|(A-{\sigma})\cap V|=n-|J|,$$
  $$|(A-{\sigma})\cap W|=n-|J|,$$
  $$|(A-{\sigma})\cap Y|=\cup_{j\notin J}Y_j,$$
  $$|(A-{\sigma})\cap L|=2n(t-\sum_{j\in J-1}t_j)+n-1-|J_1|.$$
  Thus $A-{\sigma}$  satisfies conditions in the second item. By induction,
  ${\rm Truth}(A-{\sigma})=1$. Since $\sigma$ is arbitrary chosen,
  we have ${\rm Truth}(A)=0$.
  The proof of the lemma is completed.
\end{proof}

\begin{corollary}
Let $(t_1,\cdots,t_n,t)$ be tuple of numbers, and $(X,S,O)$ the corresponding occupation game  constructed in the last definition.
Then the following statements are true.
  \begin{itemize}
    \item If there is a  subset  of $\{t_1,\cdots,t_n\}$ whose numbers adds to $t$, then ${\rm Truth}(X)=1$.
    \item If there is no  subset  of $\{t_1,\cdots,t_n\}$ whose numbers adds to $t$, then ${\rm Truth}(X)=0$.
  \end{itemize}
\end{corollary}
\begin{proof}
This follows from the last lemmas.
\end{proof}

We now prove Theorem \ref{hard}.

Let $(t_1,\cdots,t_n,t)$ be tuple of numbers, and $(X,S,O)$ the corresponding occupation game  constructed in the last definition.

By the last corollary, the Subset Sum Problem asking the existence of a subset of $\{t_1,\cdots,t_n\}$ whose numbers adds to $t$, is reduced to the computation of ${\rm Truth}(X)$ in
 $(X,S,O)$.
As the Subset Sum Problem is ${\rm NP}$-complete, we conclude that the computation of ${\rm Truth}(X)$ in
 $(X,S,O)$
is ${\rm NP}$-hard. The proof of Theorem \ref{hard} is completed.

\end{document}